\documentclass[conference,9pt]{IEEEtran}
\IEEEoverridecommandlockouts
\usepackage{cite}
\usepackage{amsfonts,amsmath,amssymb,amsthm,mathrsfs,bm,bbm}
\usepackage{algorithmic}
\usepackage{graphicx}
\usepackage{overpic}
\usepackage{textcomp}
\usepackage{xcolor}

\ifCLASSOPTIONcompsoc 
	\usepackage[caption=false,font=normalsize,labelfon t=sf,textfont=sf]{subfig} 
\else 
	\usepackage[caption=false,font=footnotesize]{subfig} \fi 
    
\newcommand{\eqdef}{:=}
\newcommand{\reqdef}{=:}
\newcommand{\rvec}[1]{\mathbbm{#1}} 		
\newcommand{\rmat}[1]{\mathbbm{#1}} 	
\newcommand{\E}{\mathsf{E}}		
\newcommand{\V}{\mathsf{Var}}			
\newcommand{\stdset}[1]{\mathbbmss{#1}}	
\newcommand{\set}[1]{\mathcal{#1}}		
\renewcommand{\vec}[1]{\mathbf{#1}}		
\newcommand{\CN}{\mathcal{CN}}			
\newcommand{\herm}{\mathsf{H}}			
\newcommand{\T}{\mathsf{T}}				

\newtheorem{lemma}{Lemma}

\newtheorem{theorem}{Theorem}

\newtheorem{assumption}{Assumption}

\newtheorem{corollary}{Corollary}

\begin{document}

\title{Team Precoding Towards Scalable Cell-free Massive MIMO Networks
\thanks{This work received partial support from the Huawei funded Chair on  Future Wireless Networks at EURECOM. E.~Bj\"ornson was supported by the Grant 2019-05068 from the Swedish Research Council.}
}

\author{\IEEEauthorblockN{Lorenzo Miretti}
\IEEEauthorblockA{\textit{Communication Systems Dept.} \\
\textit{EURECOM}\\
Sophia Antipolis, France \\
miretti@eurecom.fr}
\and
\IEEEauthorblockN{Emil Bj\"ornson}
\IEEEauthorblockA{\textit{Dept. of Computer Science} \\
\textit{KTH Royal Institute of Technology}\\
Stockholm, Sweden}
\IEEEauthorblockA{\textit{Dept. of Electrical Engineering} \\
\textit{Link\"oping University}\\
Link\"oping, Sweden \\
emilbjo@kth.se}

\and
\IEEEauthorblockN{David Gesbert}
\IEEEauthorblockA{\textit{Communication Systems Dept.} \\
\textit{EURECOM}\\
Sophia Antipolis, France \\
gesbert@eurecom.fr}
}

\IEEEspecialpapernotice{(Invited Paper)} 

\maketitle

\begin{abstract}
In a recent work, we studied a novel precoding design for cell-free networks called team minimum mean-square error (TMMSE) precoding, which rigorously generalizes centralized MMSE precoding to distributed operations based on transmitter-specific channel state information (CSI). Despite its flexibility in handling different cooperation regimes at the CSI sharing level, TMMSE precoding assumes network-wide sharing of the data bearing signals, and hence it is inherently not scalable. In this work, inspired by recent advances on scalable cell-free architectures based on user-centric network clustering techniques, we address this issue by proposing a novel version of the TMMSE precoding design covering partial message sharing. The obtained framework is then successfully applied to derive a variety of novel, optimal, and efficient precoding schemes for a user-centric cell-free network deployed using multiple radio stripes. Numerical simulations of a typical industrial internet-of-things scenario corroborate the gains of TMMSE precoding over competing schemes in terms of spectral efficiency under different power constraints. Although presented in the context of downlink precoding, the results of this paper may be applied also on the uplink.
\end{abstract}


\section{Introduction} 
Cell-free massive multiple-input multiple-output (MIMO) is one of the main candidate technologies for meeting the ambitious requirements of future wireless generation networks. Its main feature is the combination of the benefits of ultra-dense deployments and \textit{coordinated multi-point} (CoMP) methods, with the goal of offering a uniformly good quality of service to all users. Compared with classical cellular technologies, cell-free massive MIMO gives the opportunity to realize this goal in a significantly more efficient way in terms of bandwidth and energy consumption \cite{ngo2017cellfree,demir2021foundations}. 

From a theoretical perspective, cell-free massive MIMO can be interpreted as an evolution of the well-known \textit{network MIMO} \cite{karakayali2006network} and \textit{cloud radio access network} (C-RAN) \cite{checko2015cran} concepts towards a more scalable, and hence practically feasible, implementation. This crucial enhancement has been essentially enabled by the analytical framework developed within the massive MIMO literature \cite{marzetta2016fundamentals,massivemimobook}, which allowed a more refined performance analysis especially in the presence of imperfect channel state information (CSI). 

In fact, the development of cell-free massive MIMO was largely driven by the necessity of developing practical CoMP methods that do not require expensive network-wide information processing and sharing. Indeed, the early approaches recommended the use of time-division duplex (TDD) operations, and to run very simple distributed precoding/combining schemes at each access point (AP) on the basis of locally measured uplink (UL) channel samples only \cite{ngo2017cellfree}. The massive MIMO regime was then exploited to counteract the losses induced by this very limited CSI configuration. From this original idea, several improvements have been proposed involving more advanced processing and resource allocation techniques \cite{nayebi2017precoding,emil2020cellfree,buzzi2020usercentric,emil2020scalable,du2021cellfree}. The main objective of most of this subsequent literature is the definition of scalable and efficient architectures identifying which set of APs and computational units should serve a given user, with which precoding/combing scheme, and on the basis of which information.


\subsection{Optimal distributed precoding / combining}
One of the major issues behind moving from centralized to distributed processing is that limiting the CSI sharing makes the optimal precoding / combining design problem much more complicated. Therefore, previous works confined their analysis to the direct transposition of centralized methods taken from the massive MIMO literature such as \textit{maximum-ratio transmission / combining} (MRT / MRC), \textit{zero-forcing} (ZF), and \textit{minimum mean-square error} (MMSE) processing \cite{demir2021foundations}. On top of being heuristic, this approach is often not general enough, in the sense that it can be applied only to very specific CSI sharing patterns similar to the fully distributed setup in \cite{ngo2017cellfree}. This long-lasting problem has been solved only recently in \cite{miretti2021team} using the so-called \textit{team MMSE} (TMMSE) method. Building on powerful multi-agent control theoretical tools, and focusing on simple achievable rate bounds \cite{marzetta2016fundamentals,massivemimobook}, the TMMSE method provides rigorous yet practical guidelines for optimal distributed precoding  design under partial CSI sharing. As an important application, \cite{miretti2021team} derives the closed-form optimal \textit{local} precoders for the fully distributed setup studied in \cite{ngo2017cellfree}, improving upon previously known heuristics especially in the presence of pilot contamination and/or line-of-sight (LoS) components. Furthermore, \cite{miretti2021team} also derives the closed-form optimal precoders assuming that CSI is shared unidirectionally along a serial fronthaul, leading to an efficient recursive algorithm suitable for cell-free massive MIMO networks deployed using a so-called \textit{radio stripe} \cite{interdonato2019ubiquitous,shaik2020mmse}. An important remark is that, although presented in the context of downlink precoding, the results in \cite{miretti2021team} exploit the uplink-downlink (UL-DL) duality principle given by \cite{massivemimobook}, hence they also readily provide optimal UL combining schemes.

\subsection{Paper overview and summary of contributions}
The team MMSE method as exposed in \cite{miretti2021team} presents some limitations. First, it assumes that each user is served by all APs in the network, and hence it is not scalable. Second, the DL case assumes a rather unrealistic sum-power constraint, so that the UL-DL duality principle in \cite{massivemimobook} applies. In Section \ref{sec:TMMSE} of this work, we address the above limitations as follows:
\begin{itemize}
\item We show that the TMMSE method can be optimally merged with the \textit{user-centric} network clustering framework given by \cite{emil2020scalable}. This leads to a novel optimal distributed precoding / combining design jointly considering partial CSI sharing and arbitrary user-AP association rules.
\item We study the effect of a simple suboptimal technique for adapting the DL TMMSE solution to a per-TX power constraint. We argue that the resulting signal-to-noise ratio (SNR) penalty is negligible in practice, making the TMMSE method an attractive solution also under more realistic power constraints.
\end{itemize}
As a second main contribution, in Section \ref{sec:stripes}, we apply our results to an \textit{industrial internet-of-things} (IIoT) scenario with a cell-free massive MIMO system deployed using multiple radio stripes. Specifically, we derive a set of novel closed-form optimal distributed precoders under various information structures. These include the case of a fully distributed user-centric architecture as defined in \cite{buzzi2020usercentric,emil2020scalable,demir2021foundations}, but also more involved yet interesting cases. In particular, we obtain novel schemes that can be efficiently implemented using recursive algorithms along each stripe, significantly reducing the scalability issue of cell-free massive MIMO with respect to the number of APs. Note that, on top of covering scalable user-AP association rules, a relevant difference between this work and \cite{miretti2021team} is that the latter considers only a single radio stripe. The proposed schemes are finally  tested using numerical simulations.

The system model is provided in Section \ref{sec:model}. Similarly to \cite{miretti2021team}, we focus on DL precoding only, but we remark that our results can be readily applied to derive optimal UL combiners. 

\textit{Notation:} We reserve italic letters (e.g., $a$) for scalars and functions, boldface letters (e.g., $\vec{a}$, $\vec{A}$) for vectors and matrices, and calligraphic letters (e.g., $\mathcal{A}$) for sets. Random quantities are distinguished from their realizations as follows: $\rvec{a}$, $\rvec{A}$ denote random vectors and matrices; $A$ denotes a random scalar, or a generic random variable taking values in some unspecified set $\set{A}$. The operators $(\cdot)^\T$, $(\cdot)^\herm$ denote respectively the transpose and Hermitian transpose of matrices and vectors. We denote the Euclidean norm by $\|\cdot\|$, and the Frobenius norm by $\|\cdot\|_\mathrm{F}$. Given $n>2$ random matrices $\rvec{A}_1,\ldots,\rvec{A}_n$ with joint distribution $p(\vec{A}_1,\ldots,\vec{A}_n)$, we say that $\rvec{A}_1\to \rvec{A}_2 \to \ldots \to \rvec{A}_n$ forms a Markov chain if $p(\vec{A}_i|\vec{A}_{i-1},\ldots,\vec{A}_1) = p(\vec{A}_i|\vec{A}_{i-1})$ $\forall i \geq 2$. We use $\mathrm{diag}(\vec{A}_1,\ldots,\vec{A}_n)$ to denote a block-diagonal matrix with the matrices $\vec{A}_1,\ldots,\vec{A}_n$ on its diagonal. We denote by $\vec{e}_n$ the $n$-th column of the identity matrix $\vec{I}$. We use $\prod_{i=l'}^{l}\vec{A}_i := \vec{A}_{l}\vec{A}_{l-1}\ldots\vec{A}_{l'}$ for integers $l \geq l' \geq 1$ to denote the \textit{left} product chain of $l-l'+1$ ordered matrices of compatible dimension, and we adopt the convention $\prod_{i=l'}^{l} \vec{A}_i = \vec{I}$ for $l<l'$. 

\section{System Model}
\label{sec:model}
\subsection{Channel model}
Consider a network of $L$ transmitters (TXs) indexed by $\mathcal{L}:=\{1,\ldots,L\}$, each of them equipped with $N$ antennas, and $K$ single-antenna receivers (RXs) indexed by $\mathcal{K}:=\{1,\ldots,K\}$. Let an arbitrary channel use be governed by the MIMO channel law
\begin{equation}
\rvec{y} = \sum_{l=1}^L \rmat{H}_l\rvec{x}_l + \rvec{n},
\end{equation}
where the $k$-th element of $\rvec{y} \in \stdset{C}^{K}$ is the received signal at RX $k$, $\rmat{H}_l \in \stdset{C}^{K\times N}$ is a sample of a stationary ergodic random process modelling the fading between TX $l$ and all RXs, $\rvec{x}_l \in \stdset{C}^N$ is the transmitted signal at TX~$l$, and $\rvec{n}\sim \CN(\vec{0},\vec{I})$ is a sample of a white noise process. This channel model is relevant, e.g, for narrowband or wideband OFDM systems where transmission spans several realizations of the fading process. For most parts of this work, we do not specify the distribution of $\rmat{H}:=\begin{bmatrix}
\rmat{H}_1,\ldots,\rmat{H}_L
\end{bmatrix}$. However, we reasonably assume the channel submatrices corresponding to different TX-RX pairs to be mutually independent, and finite fading power $\E[\|\rmat{H}\|_{\mathrm{F}}^2]<\infty$. Furthermore, we focus on $N<K$, that is, on the regime where TX cooperation is particularly useful.

\subsection{Partial message and CSIT sharing}
In order to implement scalable cell-free massive MIMO networks, there is the need to limit the information sharing between the processing units controlling the distributed TXs.\footnote{In this work, we assume that each TX is controlled by a separate processing unit, but the results can be easily generalized to the case where multiple TXs are controlled by the same processing unit.} In particular, by focusing on two of the dominant components in the information sharing burden, we consider the following information constraints: 

\textbf{Partial message sharing:}
Let $U_k \sim \CN(0,1)$ be a sample of the i.i.d. data bearing signal for RX $k \in \set{K}$. We assume the message $U_k$ to be available only at a subset $\set{L}_k \subseteq \set{L}$ of the TXs. We do not specify how the sets $\set{L}_k$ are formed, but we remark that scalable cell-free networks are expected to implement some kind of user-centric rule \cite{demir2021foundations} allowing each RX to be granted service without relying on the notion of cell. Given the sets $\set{L}_k$, we also define 
\begin{equation}
\set{K}_l \eqdef \{k \in \set{K} \text{ s.t. } l \in \set{L}_k \}, \quad l \in \set{L},
\end{equation}
that is, the subset $\set{K}_l\subseteq \set{K}$ of RXs being served by TX $l$.

\textbf{Partial CSIT sharing:}
Consider a \textit{distributed} CSIT configuration \cite{miretti2021team}, i.e., where each TX has some local side information $S_l\in \set{S}_l$ about the channel state $\rmat{H}$. For example, $S_l$ may include local measurements of the local channel $\rvec{H}_l$,  and the output of some CSIT sharing procedure. We assume $(\rmat{H},S_1,\ldots,S_L)$ to be a sample of an ergodic stationary process with first order joint distribution fixed by nature/design, and known by all TXs. Note that, due to the looser time sensitivity, sharing statistical information is in practice less challenging than sharing estimates of $\rvec{H}$. 

\subsection{Distributed linear precoding}
Given the above information constraints, we then let each TX $l$ form its transmit signal according to the following distributed linear precoding scheme:
\begin{equation}\label{eq:distributed_precoding}
\rvec{x}_l = \sum_{k\in \set{K}_l}\sqrt{p_k}\rvec{t}_{l,k}U_k,\quad \rvec{t}_{l,k} = \rvec{t}_{l,k}(S_l),
\end{equation}
where $(p_1,\ldots,p_K) \reqdef \vec{p}\in \stdset{R}_+^K$ is a fixed power allocation vector, and where $\rvec{t}_{l,k}\in \stdset{C}^N$ for $k\in \set{K}_l$ is a linear precoder applied at TX $l$ to message $U_k$ based only on the local information $S_l$. More formally, by letting $(\Omega,\Sigma,\mathbb{P}) $ be the underlying probability space over which all random quantities are defined, we constrain $\rvec{t}_{l,k}$ for $k\in \set{K}_l$ within the vector space $\mathcal{T}_l$ of square-integrable $\Sigma_l$-measurable functions $\Omega \to \stdset{C}^N$, where $\Sigma_l\subseteq \Sigma$ denotes the sub-$\sigma$-algebra generated by $S_l$ on $\Omega$, called the \textit{information subfield} of TX $l$ \cite{yukselbook}. We finally denote the full precoding vector for message $U_k$ by $\rvec{t}_k^\T := \begin{bmatrix}
\rvec{t}_{1,k}^\T & \ldots & \rvec{t}_{L,k}^\T 
\end{bmatrix}^\T $, and let $\rvec{t}_k \in \mathcal{T}^{(k)}:=\prod_{l=1}^L\mathcal{T}_l^{(k)}$, where
\begin{equation}
\set{T}_l^{(k)} \eqdef \begin{cases} \set{T}_l & \text{if } l \in \set{L}_k, \\
\{\rvec{t}_{l,k}(S_l) = \vec{0} \; \mathrm{a.s.}\} & \text{otherwise.} \end{cases}
\end{equation}

\subsection{Performance metric}
We measure the network performance under the specified transmission scheme by using Shannon (ergodic) achievable rates predicted by the popular \textit{hardening} bound \cite{marzetta2016fundamentals,massivemimobook},
\begin{equation}\label{eq:hardening_bound}
R_k^{\mathrm{hard}} \eqdef  \log\left(1+\mathrm{SINR}_k^{\mathrm{DL}}\right),
\end{equation}
\begin{equation}\mathrm{SINR}_k^{\mathrm{DL}} \eqdef \dfrac{p_k|\E[\rvec{g}_k^\herm\rvec{t}_k]|^2}{p_k\mathrm{Var}[\rvec{g}_k^\herm\rvec{t}_k] + \sum_{j\neq k}p_j\E[|\rvec{g}_k^\herm\rvec{t}_j|^2]+1},
\end{equation}
where $\begin{bmatrix}
\rvec{g}_1 & \ldots & \rvec{g}_K \end{bmatrix} :=\rmat{H}^\herm
$. 
We then let $\mathcal{R}(P_1,\ldots,P_L)$ be the union of all simultaneously achievable rate tuples $(R_1,\ldots,R_K)\in \stdset{R}_+^K$ such that $R_k \leq R_k^{\mathrm{hard}}$ $\forall k \in \mathcal{K}$ for some set of distributed precoders $\{\rvec{t}_k \in \set{T}^{(k)}\}_{k=1}^K$ and power allocation $\vec{p}\in \set{R}_+^K$ satisfying 
\begin{equation}
\sum_{k\in\set{K}}p_k\E[\|\rvec{t}_{l,k}\|^2]\leq P_l<\infty, \quad \forall l \in \set{L}.
\end{equation}
The set $\mathcal{R}(P_1,\ldots,P_L)$ is an inner bound for the capacity region of the considered network with partial message and CSIT sharing, subject to a long-term per-TX power constraint $\E[\|\rvec{x}_l\|^2]\leq P_l$, $\forall l \in \set{L}$. The goal of this paper is to provide a method for designing distributed linear precoders spanning the largest possible region $\set{R}'\subseteq \set{R}(P_1,\ldots,P_L)$.

\section{Team MMSE Precoding}
\label{sec:TMMSE}
\subsection{UL-DL duality and the MSE criterion}\label{ssec:duality}
In order to simplify the challenging problem of optimal distributed precoding design, we first focus on a relaxed version by assuming a \textit{sum} power constraint $\sum_{l=1}^L P_l \leq P$. Specifically, we focus on 
\begin{equation}
\set{R}_{\mathrm{sum}}(P) \eqdef  \bigcup_{\sum_{l=1}^{L}P_l \leq P} \set{R}(P_1,\ldots,P_L),
\end{equation}
which is an outer bound to $\set{R}(P_1,\ldots,P_L)$ $\forall (P_1,\ldots,P_L)$ s.t. $\sum_{l=1}^{L}P_l \leq P$.
This simplification allows us to exploit the UL-DL duality principle \cite[Th.~4.8]{massivemimobook}, and obtain a convenient design criterion as established by the following result. 

\begin{theorem}\label{th:duality}
Let $\vec{w}\eqdef(w_1,\ldots,w_K)$ be a vector of weights belonging to the simplex $\set{W} \eqdef\{\vec{w} \in \stdset{R}_+^K \; | \; \sum_{k=1}^Kw_k = 1\}$, and define $\vec{W}\eqdef \mathrm{diag}(\vec{w})$. Then, all rate tuples $(R_1,\ldots,R_K)$ satisfying
\begin{equation}\label{eq:MSE_bound}
R_k\leq \log(\mathrm{MSE}(\rvec{t}_k))^{-1},
\end{equation}
\begin{equation}
\mathrm{MSE}_k(\rvec{t}_k) \eqdef \E\left[ \left\|\vec{W}^{\frac{1}{2}}\rmat{H}\rvec{t}_k - \vec{e}_k\right\|^2 +\dfrac{1}{P}\|\rvec{t}_k\|^2 \right],
\end{equation}
for some $\rvec{t}_k \in \set{T}^{(k)}$,  belong to $\set{R}_{\mathrm{sum}}(P)$. Furthermore, if $\rvec{t}_k^\star$ minimizes $\mathrm{MSE}_k(\rvec{t}_k)$ over $\set{T}^{(k)}$, then $(R_1,\ldots,R_K)$, $R_k = \log(\mathrm{MSE}(\rvec{t}^\star_k))^{-1}$ is Pareto optimal w.r.t. $\set{R}_{\mathrm{sum}}(P)$, and every tuple in the Pareto boundary of $\set{R}_{\mathrm{sum}}(P)$ is achievable for some $\vec{w} \in \set{W}$.
\end{theorem}
\begin{proof}
The proof for the special case of perfect message sharing, i.e., by considering $\set{T} \eqdef \prod_{l=1}^L\set{T}_l$ instead of $\set{T}^{(k)}$ $\forall k \in \set{K}$, is given by \cite[Th.~1]{miretti2021team}. The proof for the general case follows by observing that none of the steps in the proof of \cite[Th.~1]{miretti2021team} (and, in particular, the UL-DL duality principle) is impacted if we replace $\set{T}$ with $\set{T}^{(k)}$.
\end{proof}
Theorem~\ref{th:duality} states that, if we consider a sum power constraint, the optimal distributed precoding design can be obtained by solving (for each RX $k \in \set{K}$)
\begin{equation}\label{eq:UC-TMMSE}
\underset{\rvec{t}_k\in \set{T}^{(k)}}{\text{minimize}} \;  \mathrm{MSE}_k(\rvec{t}_k),
\end{equation}
and where the operating point on the boundary of the performance region $\set{R}_{\mathrm{sum}}(P)$ is parametrized by the $K-1$ real coefficients $\vec{w}\in \set{W}$. The operating point $\vec{w}$ should be ideally selected by solving a dual UL power allocation problem optimizing some network utility. However, in practice, $\vec{w}$ is often chosen heuristically, for instance by interpreting its element as RXs' priorities. Given $\vec{w}\in \set{W}$ and a set of distributed precoders $\{\rvec{t}_k\}_{k=1}^K \in  \set{T}^{(k)}$, the desired achievable scheme is then obtained by computing a feasible DL power allocation vector $\vec{p}$ solving the linear system of equations in \cite[Theorem~4.8]{massivemimobook}. 

Problem~\eqref{eq:UC-TMMSE} extends the so-called \textit{team} MMSE (TMMSE) precoding design criterion, studied in \cite{miretti2021team}, to the case of partial message sharing. Although more general, in this work we also refer to the solution of Problem~\eqref{eq:UC-TMMSE} as the optimal TMMSE precoders. This name originates from the specific solution approach that we adopt in this work, which is outlined in the following section.

\subsection{Team theory for distributed precoding design}
In this section we provide a useful set of necessary and sufficient conditions for optimal TMMSE precoding design. The key idea is that, by taking a user-centric perspective, we can optimize each precoder $\rvec{t}_k\in \set{T}^{(k)}$ by assuming full message sharing within a reduced network composed only by the TXs in $\set{L}_k$. Therefore, Problem \eqref{eq:UC-TMMSE} can be solved by means of the \textit{team theoretical} arguments \cite{yukselbook} developed in \cite{miretti2021team} for the full message sharing case. 

\begin{theorem}\label{th:quadratic_teams}
If $\E[\|\rmat{H}^\herm\rmat{H}\|^2_{\mathrm{F}}]<\infty$,
then Problem~\eqref{eq:UC-TMMSE} admits a unique optimal solution, which is also the unique solution $\rvec{t}_k^\star \in \set{T}^{(k)}$ satisfying a.s. ($\forall l \in \mathcal{L}^{(k)}$)
\begin{equation}\label{eq:UC_stationary}
\E[\rmat{O}_{l,l}|S_l]\rvec{t}^\star_{l,k}(S_l) + \sum_{j \in \set{L}^{(k)}\backslash \{l\}}\E[\rmat{O}_{l,j}\rvec{t}^\star_{j,k}|S_l] = \E[\rvec{g}_{l,k}|S_l]\vec{W}^{\frac{1}{2}}, 
\end{equation}
where $\rmat{O}_{l,l} \eqdef \rmat{H}_l^\herm\vec{W}\rmat{H}_l + P^{-1}\vec{I}$, $\rmat{O}_{l,j} \eqdef \rmat{H}_l^\herm\vec{W}\rmat{H}_j$ $\forall l\neq j$, and $\rvec{g}_{l,k}$ is the subvector of $\rvec{g}_k$ corresponding to the channel of TX $l$.
\end{theorem}
\begin{proof}
We rewrite the objective of \eqref{eq:UC-TMMSE} by replacing $\rmat{H}$ with a reduced channel matrix $\rmat{H}^{(k)}$ containing only the columns related to the TXs in $\mathcal{L}_k$, and $\rvec{t}_k$ with a corresponding reduced distributed precoding vector $\rvec{t}_k^{(k)} \in \prod_{l\in\mathcal{L}_k}\set{T}_l$, yielding
\begin{equation}
\E\left[ \left\|\vec{W}^{\frac{1}{2}}\rmat{H}^{(k)}\rvec{t}_k^{(k)} - \vec{e}_k\right\|^2 +\dfrac{1}{P}\|\rvec{t}_k^{(k)}\|^2 \right].
\end{equation}
We then notice that the minimization problem over $\rvec{t}_k^{(k)} \in \prod_{l\in\mathcal{L}_k}\set{T}_l$ has exactly the same form as the TMMSE precoding design problem studied in \cite{miretti2021team}. Therefore, the results of \cite{miretti2021team} based on the theory of \textit{quadratic} teams \cite{yukselbook} apply, and the proof readily follows.
\end{proof}

The optimality conditions \eqref{eq:UC_stationary} correspond to an infinite dimensional linear system of equations, which can be  solved via one of the many approximation methods available in the literature, such as the ones surveyed in \cite{yukselbook}. Importantly, \eqref{eq:UC_stationary} does not assume any particular CSIT acquisition scheme, and it only requires the mild condition $\E[\|\rmat{H}^\herm\rmat{H}\|^2_{\mathrm{F}}]<\infty$ on the fading distribution, which is satisfied, e.g., for all physically consistent distributions with bounded support, or for the classical Gaussian fading model. In the second part of this work, we focus on special yet relevant cases where \eqref{eq:UC_stationary} can be solved analytically. To this end, we first rewrite \eqref{eq:UC_stationary} under two additional assumptions, which are consistent with the canonical cell-free massive MIMO paradigm \cite{demir2021foundations}.

\begin{assumption}[Local channel estimation]\label{ass:TDD}
For every $l\in \mathcal{L}$, let $\hat{\rmat{H}}_{l}$ be the estimate of the local channel $\rmat{H}_l$ available at TX $l$, and $\rmat{E}_l:= \rmat{H}_l - \hat{\rmat{H}}_{l}$ be the local  estimation error. Assume that $\hat{\rmat{H}}_{l}$ and $\rmat{E}_l$ are independent. Furthermore, assume $\E[\rmat{E}_l] = \vec{0}$, and that $\E[\rmat{E}_l^\herm\vec{W}\rmat{E}_l] =: \vec{\Psi}_l$ has finite elements. Finally, assume that $(\hat{\rmat{H}}_{l},\rmat{E}_{l})$ and $(\hat{\rmat{H}}_{j},\rmat{E}_{j})$ are independent for $l\neq j$.
\end{assumption} 

\begin{assumption}[CSIT sharing mechanism]\label{ass:CSIT_sharing}
For every $(l,j) \in \mathcal{L}^2$ s.t. $l\neq j$, assume the following Markov chain: 
\begin{equation}
\rmat{H}_l \to \hat{\rmat{H}}_{l} \to S_l \to S_j \to \hat{\rmat{H}}_{j} \to \rmat{H}_j.
\end{equation}
\end{assumption}
Assumption \ref{ass:TDD} holds, e.g., for pilot-based MMSE estimates of Gaussian channels exploiting channel reciprocity in time-division duplex (TDD) systems \cite{demir2021foundations}.  Assumption \ref{ass:CSIT_sharing} essentially states that all the  available information about the local channel $\rmat{H}_l$ is fully contained in $S_l$ at TX $l$, and that TX $j$ can only obtain a degraded version of it through some arbitrary CSIT sharing mechanism. 

\begin{lemma}\label{lem:stationarity_imperfect}
Suppose the assumption of Theorem~\ref{th:quadratic_teams}, Assumption~\ref{ass:TDD}, and Assumption~\ref{ass:CSIT_sharing} hold. Then, the unique solution to Problem~\eqref{eq:UC-TMMSE} is given by the unique $\rvec{t}_k^\star \in \set{T}^{(k)}$ satisfying ($\forall l \in \mathcal{L}_k$)
\begin{equation}\label{eq:stationary_imperfect}
\rvec{t}_{l,k}^\star(S_l) =  \rmat{T}_l\left(\vec{e}_k-\sum_{j \in \set{L}_k\backslash  \{l\}} \vec{W}^{\frac{1}{2}}\E\left[\hat{\rmat{H}}_{j}\rvec{t}^\star_{j,k}\Big|S_l\right] \right) \quad \mathrm{a.s.}, 
\end{equation} 
where $\rmat{T}_l:=\left(\hat{\rmat{H}}_{l}^\herm\vec{W}\hat{\rmat{H}}_{l}+\vec{\Psi}_l + P^{-1}\vec{I}\right)^{-1}\hat{\rmat{H}}_{l}^\herm\vec{W}^{\frac{1}{2}}$.
\end{lemma}
\begin{proof}
The proof follows from simple manipulations of \eqref{eq:UC_stationary}. The details are similar to the full message sharing case given by \cite[Lem.~2]{miretti2021team}, hence they are omitted.
\end{proof}
  
\subsection{Per-TX power constraint}
The main drawback of the proposed TMMSE precoding design is that it relies on the UL-DL duality principle for fading channels given by \cite{massivemimobook}, which requires a sum power constraint. However, for deterministic channels, a weaker form of the UL-DL duality principle and of the Pareto boundary parametrization exists also under a per-TX power constraint (see, e.g., \cite{yu2007transmitter}). Although we do not rule out the possibility of establishing a similar result for fading channels, in the following we approximate the optimal distributed precoders spanning $\mathcal{R}(P_1,\ldots,P_L)$ by using a much simpler suboptimal approach. 

Suppose $\rvec{x}_l^\star$ is the $l$-th TX signal obtained using optimal TMMSE precoding under a sum power constraint $\sum_{l=1}^LP_l = P$, and consider the trivial adaptation $(\forall l \in \set{L})$
\begin{equation}\label{eq:power_scaling}
\rvec{x}_l = \frac{1}{\nu} \rvec{x}_l^\star, \quad \nu^2 =  \max\left(1,\frac{\E[\|\rvec{x}^\star_1\|^2]}{P_1},\ldots,\frac{\E[\|\rvec{x}^\star_L\|^2]}{P_L}\right),
\end{equation}
which scales everything down until the per-TX power constraint is satisfied. In terms of performance, if $\{\rvec{t}^\star_k\}_{k=1}^K$ is the set of TMMSE precoders with corresponding power allocation $\{p_k^\star\}_{k=1}^K$ generating $\rvec{x}_l^\star$, the above method produces achievable rates given by $(\forall k \in \set{K})$
\begin{equation}
R_k = \log\left(1+ \dfrac{p^
\star_k|\E[\rvec{g}_k^\herm\rvec{t}^\star_k]|^2}{p^\star_k\V[\rvec{g}_k^\herm\rvec{t}^\star_k] + \sum_{j\neq k}p^\star_j\E[|\rvec{g}_k^\herm\rvec{t}^\star_j|^2]+\nu^2}\right), 
\end{equation}
that is, compared to the optimal solution assuming a sum power constraint, there is a SNR loss $\nu^2$ proportional to the largest violation of the TX power constraints. Nevertheless, this loss may be marginal if the interference terms are still dominating the denominator in the rate expressions. Therefore, we argue that the simple power scaling method described above may be particularly suitable for the setup of this work, where partial message and CSIT sharing indeed often induce an interference limited regime.\footnote{The saturation of $R_k$ w.r.t. the SNR is due to the simple point-to-point coding scheme achieving \eqref{eq:hardening_bound} which treats interference as noise. The extension of our results to more advanced coding schemes based, e.g., on rate-splitting and (partial) interference decoding \cite{clerckx2016rate} is an interesting future direction.}

Other variations may be also considered, for example we can repeat the TMMSE precoders computation for decreasing values of $P$ until the per-TX power constraint is satisfied. Furthermore, power scaling coefficients coupled with clipping techniques can be also used to adapt long-term power constraints to short-term power constraints of the type $\|\rvec{x}_l\|^2\leq P_l$ almost surely. We leave further discussions on power constraints for future work. 

\section{Case Study: Network of Radio Stripes}
\label{sec:stripes}
\subsection{Network architecture}
We consider an indoor IIoT scenario similar to \cite{interdonato2019ubiquitous} and depicted in Fig.~\ref{fig:IIoT}, where $L=100$ single-antenna TXs are arranged over a rectangular area (e.g., the ceiling of a warehouse) of $100 \times 50$ [$\mathrm{m}^2$] using $Q=5$ radio stripes of $M=20$ TXs each. For ease of exposition, we map all indexes $l \in \set{L}$ to pairs of indexes $(q,m) \in \set{Q} \times \set{M}$, where $\set{Q} \eqdef  \{ 1,\ldots, Q \}$ and $\set{M} \eqdef \{1,\ldots, M \}$. For example, we denote the channel  $\rmat{H}_l$ between all RXs and the $m$-th TX of the $q$-th radio stripe by $\rmat{H}_{q,m}$. We assume each stripe $q$ is controlled by a separate master processing unit placed at the edge next to TX $(q,1)$, providing coded and modulated data bearing signals $U_k$ to all the TXs along the stripe. To limit the information processing and sharing burden, we let each RX $k$ be served only by the TXs belonging to the subset $\set{Q}_k \subset \set{Q}$ of its $Q_k = 2$ closest radio stripes. This implies that the message for RX $k$ is shared only within the master units controlling the $Q_k$ stripes in $\set{Q}_k$. Similarly, we assume that, at each TX $(q,m)$, the channel $\rvec{g}_{q,m,k}$ towards RX $k$ is perfectly estimated if $q \in \set{Q}_k$, and completely unknown otherwise.  Then, we study the following CSIT sharing patterns:
\begin{enumerate}
\item No CSIT sharing: $S_{q,m} = \hat{\rvec{H}}_{q,m}$;
\item Unidirectional CSIT sharing: $S_{q,m} = (\hat{\rvec{H}}_{q,1},\ldots,\hat{\rvec{H}}_{q,m})$;
\item Bidirectional CSIT sharing: $S_{q,m} = (\hat{\rvec{H}}_{q,1},\ldots,\hat{\rvec{H}}_{q,M})$.
\end{enumerate}
For all the above cases, CSIT sharing is limited within each stripe. 
\begin{figure}[ht]
\centering
\includegraphics[trim=0.5cm 0 0.5cm 1.2cm, clip,width=0.8\columnwidth]{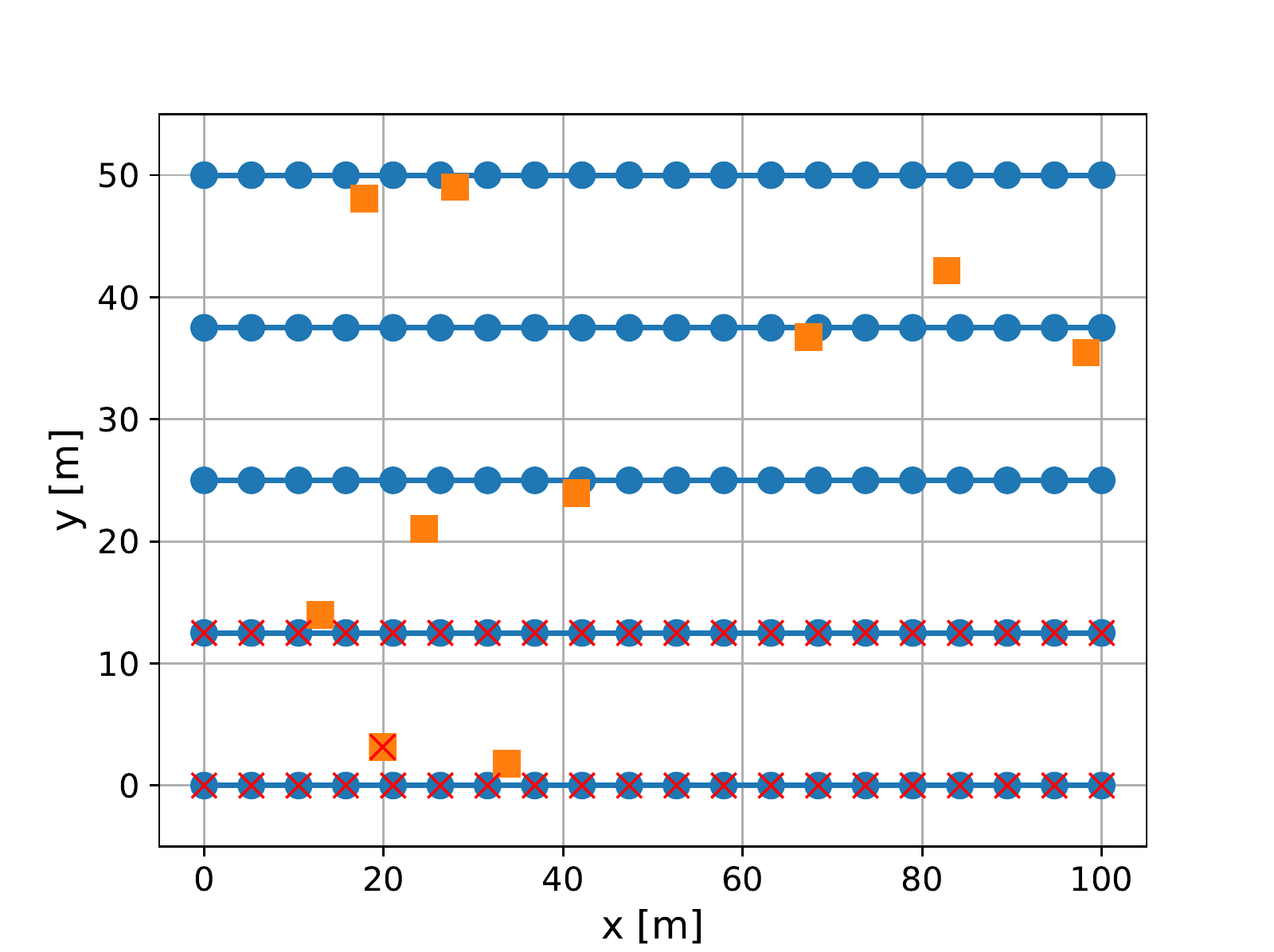}
\caption{Pictorial representation of an IIoT setup with $K=10$ RXs uniformly distributed within a rectangular service area, and a grid of $L=100$ single-antenna TXs arranged over $Q=5$ radio stripes of $M=20$ TXs each. The red crosses identify the set $\set{Q}_k$ of TXs serving an arbitrary RX, obtained from the $Q_k=2$ closest radio stripes.}
\label{fig:IIoT}
\end{figure}

\subsection{Optimal distributed precoding and recursive implementations}
The optimal solution to \eqref{eq:UC-TMMSE} can be obtained in closed form for all the considered information structures. We start with the most involved case, that is, undirectional CSIT sharing: 
\begin{theorem}\label{th:unidirectional}
The TMMSE precoder applied at the $m$-th TX of the $q$-th radio stripe to the message for RX $k$ under unidirectional CSIT sharing is given by
\begin{equation}\label{eq:TMMSE_stripes}
\rvec{t}_{q,m,k}(S_{q,m}) = \rmat{T}_{q,m}\rmat{V}_{q,m}\prod_{n=1}^{m-1} \bar{\rmat{V}}_{q,n}\vec{c}_{q,k}, \quad \forall q \in \set{Q},
\end{equation}
where we use the following short-hands:
\begin{itemize}
\item $\rmat{V}_{q,m} \eqdef (\vec{I}-\vec{\Pi}_{q,m}\rmat{P}_{q,m})^{-1}(\vec{I}-\vec{\Pi}_{q,m})$;
\item $\bar{\rmat{V}}_{q,m} \eqdef \vec{I}-\rmat{P}_{q,m}\rmat{V}_{q,m}$;
\item $\rmat{P}_{q,m} \eqdef \vec{W}^{\frac{1}{2}}\hat{\rmat{H}}_{q,m}\rmat{T}_{q,m}$;
\item $\vec{\Pi}_{q,m} \eqdef \E\left[\rmat{P}_{q,m+1}\rmat{V}_{q,m+1}\right] + \vec{\Pi}_{q,m+1}\E\left[\bar{\rmat{V}}_{q,m+1}\right]$;
\end{itemize}
and where we let $\vec{\Pi}_{q,M} \eqdef \vec{0}_{K\times K}$, and $\{\vec{c}_{q,k}\}_{q=1}^Q$ be the unique solution of the linear system of equations
\begin{equation}
\begin{cases}\vec{c}_{q,k} + \sum_{j \in \set{Q}_k \backslash \{q\}}\vec{\Pi}_{j,0} \vec{c}_{j,k} = \vec{e}_k & \forall q \in \set{Q}_k, \\
\vec{c}_{q,k} = \vec{0}_{K\times 1} & \text{otherwise.}
\end{cases}
\end{equation}
\end{theorem}
\begin{proof}
We need to verify that the proposed solution satisfies the optimality conditions given by Lemma \ref{lem:stationarity_imperfect}. Replacing \eqref{eq:TMMSE_stripes} in \eqref{eq:stationary_imperfect}, after some simple manipulations, we observe that it suffices to verify ($\forall (q,m) \in \set{Q}_k \times \set{M}$)  
\begin{equation}\label{eq:stationarity_stipes}
\begin{split}
&\rmat{V}_{q,m}\prod_{n=1}^{m-1} \bar{\rmat{V}}_{q,n}\vec{c}_{q,k} + \sum_{l \neq m} \E\left[ \rmat{P}_{q,l}\rmat{V}_{q,l}\prod_{n=1}^{l-1} \bar{\rmat{V}}_{q,n}\vec{c}_{q,k} \middle| S_{q,m}\right] \\
& + \sum_{j \in \set{Q}_k \backslash \{q\}} \sum_{l=1}^M\E\left[ \rmat{P}_{j,l}\rmat{V}_{j,l}\prod_{n=1}^{l-1} \bar{\rmat{V}}_{j,n}\vec{c}_{j,k} \middle| S_{q,m}\right] = \vec{e}_k.
\end{split}
\end{equation}
The first and second term in \eqref{eq:stationarity_stipes} correspond to the contribution of the $q$-th radio stripe, seen by TX $(q,m)$. Using the result in \cite[Th.~5]{miretti2021team} for a single radio stripe, which exploit in its proof the recursive structure of the equations and of the information structure at hand, it can be shown that 
\begin{align}
\rmat{V}_{q,m}\prod_{n=1}^{m-1} \bar{\rmat{V}}_{q,n} + \sum_{l \neq m} \E\left[ \rmat{P}_{q,l}\rmat{V}_{q,l}\prod_{n=1}^{l-1} \bar{\rmat{V}}_{q,n} \middle| S_{q,m}\right] = \vec{I}_K,
\end{align}
which simplifies \eqref{eq:stationarity_stipes} to
\begin{equation}
\vec{c}_{q,k} + \sum_{j \in \set{Q}_k \backslash \{q\}} \sum_{l=1}^M\E\left[ \rmat{P}_{j,l}\rmat{V}_{j,l}\prod_{n=1}^{l-1} \bar{\rmat{V}}_{j,n}\vec{c}_{j,k} \middle| S_{q,m}\right] = \vec{e}_k.
\end{equation}
The second term in the above equation corresponds to the contributions of the other radio stripes $j \neq q$. Since $S_{q,m}$ contains no information about the channels at the other radio stripes $j\neq q$,  the optimality conditions can be further simplified to
\begin{equation}
\vec{c}_{q,k} + \sum_{j \in \set{Q}_k \backslash \{q\}} \E\left[ \sum_{l=1}^M\rmat{P}_{j,l}\rmat{V}_{j,l}\prod_{n=1}^{l-1} \bar{\rmat{V}}_{j,n} \right]\vec{c}_{j,k} = \vec{e}_k.
\end{equation}
Following again the same approach as in \cite[Th.~5]{miretti2021team} for a single radio stripe, it can be shown via recursive arguments that 
\begin{align}
&\E\left[\sum_{l=1}^M \rmat{P}_{j,l}\rmat{V}_{j,l}\prod_{n=1}^{l-1} \bar{\rmat{V}}_{j,n} \right] \\
&= \E\left[ \rmat{P}_{j,1}\rmat{V}_{j,1} \right] + \E\left[\sum_{l=2}^M \rmat{P}_{j,l}\rmat{V}_{j,l}\prod_{n=2}^{l-1} \bar{\rmat{V}}_{j,n} \right]\E\left[\bar{\rmat{V}}_{j,1} \right]\\
&= \vec{\Pi}_{j,0}.
\end{align}
The proof is concluded by showing that the resulting system of equations $\vec{c}_{q,k} + \sum_{j \in \set{Q}_k \backslash \{q\}}\vec{\Pi}_{j,0} \vec{c}_{j,k} = \vec{e}_k$, $\forall q \in \set{Q}_k$, always has a unique solution. This can be done by following the same lines as in the proof of \cite[Th.~4]{miretti2021team}.
\end{proof}
Interestingly, \eqref{eq:TMMSE_stripes} can be implemented through a recursive procedure from the $q$-th master unit to TX $(q,M)$ involving the sequential update and forward of a $K$-dimensional complex vector $\rvec{u}^{(q,m)} = \bar{\rmat{V}}_{q,m}\rvec{u}^{(q,m-1)}$, where $\rvec{u}^{(q,0)}\eqdef \sum_{k\in \set Q_k} \vec{c}_{q,k}U_k$ is a vector of statistically precoded messages computed at the $q$-th master unit using $\{\vec{\Pi}_{q,0}\}_{q\in \set{Q}_k}$ for all the RXs $k$ that it is serving. Furthermore, to compute $\bar{\rmat{V}}_{q,m}$, each TX $(q,m)$ needs only the local estimate $\hat{\rmat{H}}_{q,m}$ and $\vec{\Pi}_{q,m}$. The required statistical information can be easily tracked through another recursive procedure involving the sequential update and forward of $\vec{\Pi}_{q,m} \in \stdset{C}^{K\times K}$, in the reverse direction. This leads to a very efficient implementation which requires only some local information exchange scaling with the number of RXs $K$. 

The bidirectional CSIT sharing case can be readily obtained from the proof of Theorem~\ref{th:unidirectional} by omitting the expectation in the summands over $l\neq m$, since $S_{q,m}$ makes these terms deterministic.
\begin{corollary}
The TMMSE precoder applied at the $m$-th TX of the $q$-th radio stripe to the message for RX $k$ under bidirectional CSIT sharing is given by \eqref{eq:TMMSE_stripes}, with $\vec{\Pi}_{q,m}$ replaced by $\bar{\rvec{P}}_{q,m} \eqdef \rmat{P}_{q,m+1}\rmat{V}_{q,m+1} + \bar{\rvec{P}}_{q,m+1}\bar{\rmat{V}}_{q,m+1}$ in the computation of $\rvec{V}_{q,m}$.
\end{corollary}
Similarly to \eqref{eq:TMMSE_stripes}, the above scheme can be also implemented using a recursive procedure. The main difference is that, in contrast to $\vec{\Pi}_{q,m}$, $\bar{\rvec{P}}_{q,m}$ needs to be computed for every channel realization. Note that all $\vec{\Pi}_{q,0}$ still need to be acquired for the computation of the statistical precoders $\vec{c}_{q,k}$.

We finally consider the no CSIT sharing case. This case is better covered by remapping $\set{Q}\times \set{M}$ to the original index set $\set{L}$, and defining the sets $\set{L}_k$ from $\set{Q}_k\times \set{M}$ accordingly. 
\begin{theorem}
The TMMSE precoder applied at the $l$-th TX to the message for RX $k$ under no CSIT sharing $S_l = \hat{\rmat{H}}_{l}$ is given by
\begin{equation}\label{eq:TMMSE_local}
\rvec{t}_{l,k}(S_l) = \rmat{T}_l\vec{c}_{l,k}, \quad \forall l \in \set{L},
\end{equation}
where we define $\vec{\Pi}_l \eqdef \E\left[\vec{W}^{\frac{1}{2}}\hat{\rmat{H}}_l\rmat{T}_l\right] $ and let $\{\vec{c}_{l,k}\}_{l=1}^L$ be the unique solution of the linear system of equations
\begin{equation}
\begin{cases}\vec{c}_{l,k} + \sum_{j \in \set{L}_k \backslash \{l\}}\vec{\Pi}_j \vec{c}_{j,k} = \vec{e}_k & \forall l \in \set{L}_k, \\
\vec{c}_{l,k} = \vec{0}_{K\times 1} & \text{otherwise.}
\end{cases}
\end{equation}
\end{theorem}
\begin{proof}
The proof follows by verifying that \eqref{eq:TMMSE_local} satisfies the optimality coditions \eqref{eq:stationary_imperfect}. The details are similar to the full message sharing case given by \cite[Th.~4]{miretti2021team}, hence they are omitted.
\end{proof}
As for \cite[Th.~4]{miretti2021team}, it can be seen that under certain conditions typically modeling non line-of-sight (NLoS) channels and no pilot contamination, \eqref{eq:TMMSE_local} coincides with the known \textit{local} MMSE solution $\rvec{t}_{l,k}(S_l) = c_{l,k}\rmat{T}_l\vec{e}_k$ with optimized \textit{large-scale fading} coefficients $c_{l,k}$ given, e.g., in \cite{demir2021foundations}. However, \eqref{eq:TMMSE_local} outperforms local MMSE precoding in the general case. Note that \eqref{eq:TMMSE_local} can also be used to derive an equivalent scheme for bidirectional CSIT sharing, by simply interpreting all the TXs of stripe $q$ as a single TX with $NM$ antennas, and assume no CSIT sharing. However, this interpretation does not easily lead to a recursive implementation, and it is more suitable for a setup where each master unit is responsible for the computation of all the precoders for its stripe, instead of distributing the task.

\begin{figure*}[!ht] 
\centerline{
\subfloat[]{\includegraphics[trim=0.5cm 0 0.5cm 1.2cm, clip, width = 0.45\linewidth]{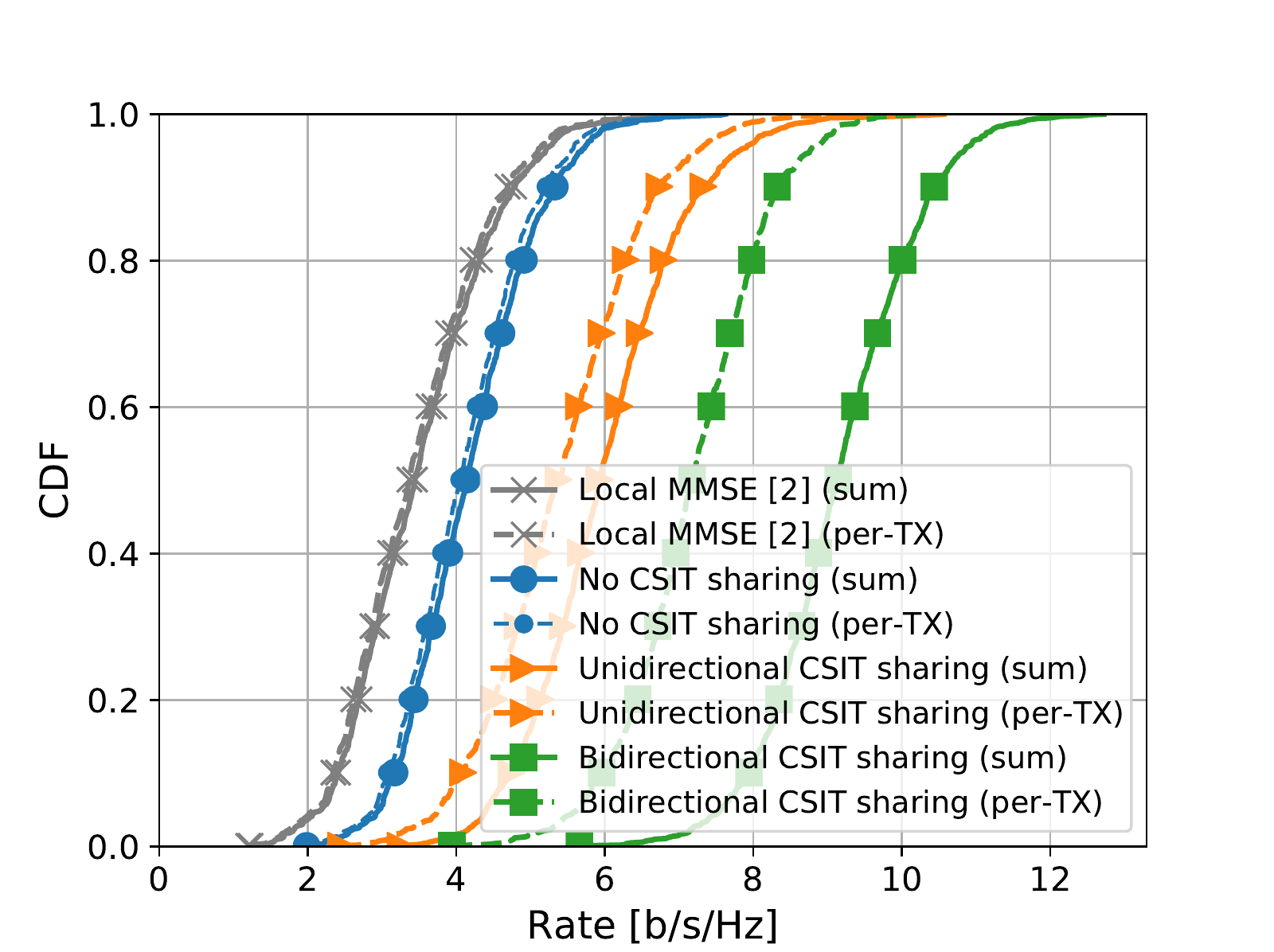} \label{fig_first_case}} 
\hfil 
\subfloat[]{\includegraphics[trim=0.5cm 0 0.5cm 1.2cm, clip, width = 0.45\linewidth]{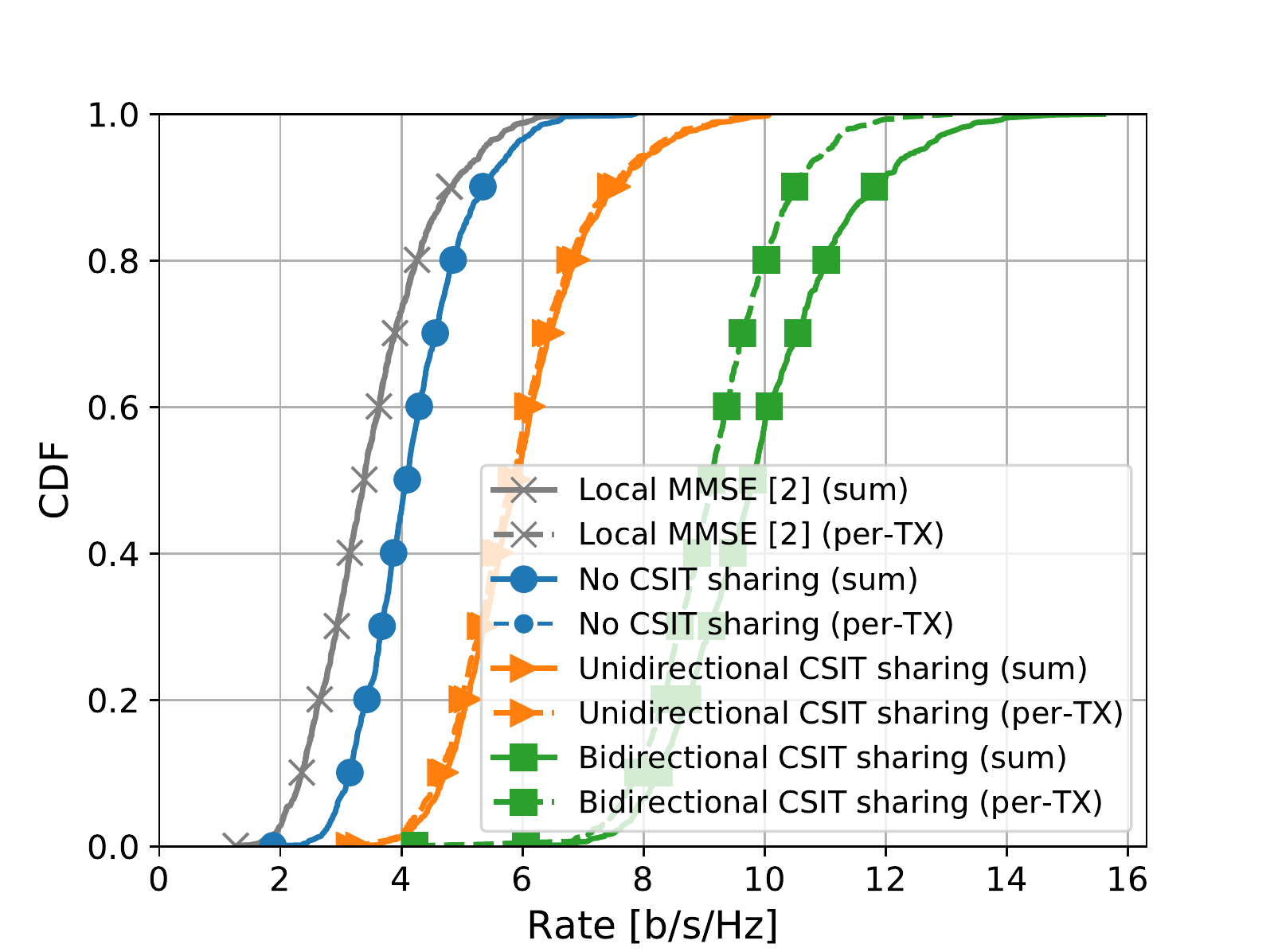} \label{fig_second_case}}} 
\caption{CDF of the rate achieved by TMMSE precoding under different CSIT sharing patterns, and by assuming a symmetric per-TX power constraint (a) $P_l = 1$ mW, and (b) $P_l = 10$ mW, $\forall l \in \set{L}$. As expected, sharing CSIT leads to performance gains, and the unidirectional case strikes an interesting compromise between the considered examples. Furthermore, the sum-power upper bound is approached as the system becomes interference limited, e.g., for lower interference suppression capabilities and/or higher power budget. Finally, the local MMSE baseline is outperformed by the competing TMMSE solution under no CSIT sharing, due to the presence of LoS components ($\kappa = 6$).} \label{fig:sim} 
\end{figure*} 

\subsection{Numerical simulations}
We let each channel coefficient $H_{l,k}$ between TX~$l$ and RX~$k$ be independently distributed as $H_{l,k} \sim \CN\left(\sqrt{\frac{\kappa}{\kappa+1}\rho^2_{l,k}}, \frac{1}{\kappa+1}\rho^2_{l,k}\right)$, where $\kappa$ denotes a common Ricean factor, and $\rho^2_{l,k}$ denotes the channel gain between TX $l$ and RX $k$. We follow the 3GPP NLoS-DH path-loss model for IIoT applications \cite{jiang2021iiot}
\begin{equation}
\mathrm{PL}_{l,k} = 21.9 \log_{10}\left(\dfrac{d_{l,k}}{1 \; \mathrm{m}}\right) + 33.6 + 20\log_{10}\left(\dfrac{f_c}{1 \; \mathrm{GHz}}\right) + Z_{l,k} \; [\text{dB}],
\end{equation}
where $f_c = 4.9$ GHz is the carrier frequency, $d_{l,k}$ is the distance between TX $l$ and RX $k$ including a difference in height of $7$ m, and $Z_{l,k}\sim \CN(0,\sigma^2)$ is an independent shadow fading term with standard deviation $\sigma = 4$. We let the noise power at all RXs be given by
$
P_{\mathrm{noise}} = -174 + 10 \log_{10}(B) + F$ dBm,
where $B = 100$ MHz is the system bandwidth, and $F = 7$ dB is the noise figure. Finally, we let $\rho^2_{l,k} := 10^{-\frac{\mathrm{PL}_{l,k} +P_{\mathrm{noise}}}{10}}$ $\text{mW}^{-1}$, and $\kappa = 6$. 

Fig.~\ref{fig:sim} reports the empirical cumulative density function (CDF) of the rates \eqref{eq:hardening_bound} achieved by TMMSE precoding under the considered CSIT sharing patterns, for $100$ realizations of the RX positions. Each RX position is independently and uniformly drawn within the considered service area. We study both the sum-power upper bound and the suboptimal adaptation \eqref{eq:power_scaling} to a per-TX power constraint. We focus on the boundary point of $\set{R}_{\mathrm{sum}}(P)$ given by $\vec{w} = \vec{1}/K$. For the sum-power constraint, the DL power allocation $\vec{p}$ is obtained using the UL-DL principle \cite{massivemimobook}, by interpreting $\vec{w}$ as  dual UL powers. For the per-TX power constraint, $\vec{p}$ is further scaled by $\nu^2$. As a baseline, we also report the performance of local MMSE precoding with optimized large-scale fading coefficients \cite{demir2021foundations}.  

The first remarkable observation is that, while having roughly the same information sharing overhead on the order of $K$ complex symbols per channel use and stripe span, the TMMSE solution under unidirectional CSIT sharing significantly improves upon the no CSIT sharing case. Therefore, it can be seen as a promising intermediate solution between no CSIT sharing and bidirectional CSIT sharing. However, this comes at the expense of additional computational complexity at each TX to implement the recursive routine. 

A second observation is that, as expected, the approach of scaling the TMMSE solution to satisfy the per-TX power constraint is basically optimal whenever the dominant performance limit is the residual interference. This is evident for the schemes with lower interference suppression capabilities, and/or for higher power budgets. However, the gap w.r.t. the sum-power upper bound increases for low powers and for the better performing information structures. In this case, further research is needed to understand whether this gap is close to the fundamental limit or only an artifact of the suboptimal approach. 

Finally, we observe that the baseline local MMSE method is not optimal under no CSIT sharing. This is motivated by the fact that, similarly to other local methods such as MRT, local MMSE does not handle well the interference originating from the channel mean. 

\bibliographystyle{IEEEbib}
\bibliography{refs}

\end{document}